\documentclass[10pt]{article}

\usepackage{amsmath}
\usepackage{amsfonts}
\usepackage{amssymb}
\usepackage{amsthm}
\usepackage{stmaryrd} 
\usepackage{braket}
\usepackage{empheq}
\usepackage{graphicx}
\usepackage{subfigure}
\usepackage[round]{natbib}
\usepackage{footmisc}
\usepackage{appendix}
\usepackage{epsf}
\usepackage{epsfig}
\usepackage{tikz,tikz-3dplot}
\usetikzlibrary{decorations.markings,decorations.pathmorphing}
\usepackage{caption}
\usepackage{pgfplots}
\usepackage[utf8]{inputenc}
\usepackage[T1]{fontenc}
\usepackage[english]{babel}
\usepackage{sectsty}
\usepackage{hyperref}
\usepackage{authblk}

\subsubsectionfont{\normalfont\itshape}
\usepackage{fullpage}

\renewcommand{\leq}{\leqslant}

\DeclareMathOperator{\tr}{tr}

\theoremstyle{plain}
\newtheorem{theorem}{Theorem}[section]

\newtheorem*{theorem*}{Theorem}

\theoremstyle{definition}

\newtheorem{remark}[theorem]{Remark}

\title{An interpretation-independent formulation\\of the measurement problem}
\date{}
\author[1 2]{\textsc{Antoine Soulas\thanks{\texttt{antoine.soulas@univie.ac.at}} }}
\affil[1]{Quantum Optics, Quantum Nanophysics and Quantum Information, Faculty of Physics, University of Vienna, Vienna, Austria \vspace{0.2cm}} 
\affil[2]{Institute for Quantum Optics and Quantum Information (IQOQI), Austrian Academy of Sciences, Vienna, Austria \vspace{0.2cm}}

\begin{document}
\maketitle

\vspace{-1.5cm}

\begin{center}
  \rule{6cm}{1pt}
\end{center}

\abstract{In this paper, we do not try to solve the measurement problem, but rather to properly formulate it. One of the reasons why it still lacks a precise, agreed definition is that the problem may take very different aspects depending on the interpretation of quantum mechanics embraced. Inspired by the methodology of theory-independent results like Bell’s theorem, we propose to identify the common root of the puzzle in an interpretation-independent way \textit{i.e}. as a property of the empirical statistics only, before deriving its philosophical consequences. The key observation is that quantum matter can not be described by a Kolmogorovian probabilistic theory. Arguing that the Kolmogorov axioms of probability theory are the postulates of epistemic uncertainty leads us to reformulate the measurement problem as the impossibility to build an ontology independent of epistemology for quantum matter. Said differently, there exists no God's-eye view on quantum systems. Although it is meaningless to solve the measurement problem as defined in this way, insofar as it is a feature of the universe's statistics, such a formulation may on the other hand bring side benefits. In particular, we argue that it allows to: (i) shed a new light on the variety of interpretations; (ii) propose a fundamental reason why quantum mechanics and general relativity are so incompatible, not relying on purely mathematical or technical arguments; (iii) guide the quest for quantum gravity.}

\vspace{0.4cm}
\hrule

\tableofcontents

\section{Introduction} \label{sec:intro}

Despite a century of controversy, the quantum measurement problem remains as far as it gets from a consensual solution within the quantum foundations community; see for instance \cite{schlosshauer2011elegance}'s compilation of the views of 17 physicists and philosophers on the topic, or the broader surveys \citep{schlosshauer2013snapshot, sivasundaram2016surveying, shmahal2025does}. The theory of decoherence and several no-go theorems have certainly provided some important new insights, but have not ended the debates — quite the contrary. An important difficulty is that the problem still lacks a precise, universally accepted definition that would constitute a starting point from which to discuss. Indeed, authors generally don't spend much time defining it, but rather quickly embark in an explanation of their own solution or related thoughts. Moreover, the puzzle may manifest itself in various ways depending on the interpretation of quantum mechanics (QM) embraced, leading to divergent opinions, not only on what can be considered as an acceptable answer, but also on the very question to address \citep{rovelli1996relational, schlosshauer2011elegance, muller2023six}. This lead \cite{feynman1982simulating} to admit: ‘I cannot define the real problem, therefore I suspect there's no real problem, but I'm not sure there's no real problem’. That's why, before even trying to solve the measurement problem (whatever this might mean), it seems reasonable to first properly formulate what the issue is, and to do so in an interpretation-independent way.

Upon reviewing the dedicated literature, it turns out that, although the proposed solutions may differ widely, almost all definitions of the problem can roughly be cast into two categories. The most frequently encountered one expresses it as the conflict between unitary evolution and collapse \citep{everett1957relative, peres1986quantum, albert2009quantum, griffiths2013consistent, bacciagaluppi2013measurement, zwirn2014decoherence, gisin2018collapse, barrett2019quantum, mermin2022there, cuffaro2023measurement}. This can also include authors like \cite{bell1990against} and \cite{brukner2017quantum} relating the issue to the ill-defined notion of ‘measurement’, as it is clear that they too have in mind the absence, in textbook QM, of a clear rule prescribing when to apply either one or the other postulate. The predominance of this definition in the literature probably owes to the major historical influence of \cite{von1932mathematical}'s book in shaping standard paradigms for quantum physics, including his classic presentation of the measurement problem (section VI.1. in the reference). However, this century-old textbook formulation can not be taken as a universal definition. Indeed, most of the currently existing interpretations of QM do not even have a notion of collapse in the first place: had physicists first formulated quantum theory in the terms of Bohmian mechanics or many-worlds for instance, the problem would never have been phrased this way. Furthermore, QM is a probabilistic theory, hence it is absolutely normal that it involves an update of the probabilities when new information is acquired. The coexistence of such a discontinuous epistemic update together with a law governing the evolution of the probabilities is a natural feature of any dynamical probability theory, as was already emphasised by \cite{brukner2017quantum, mermin2022there} and many others.

A second category of authors express the problem as the puzzle that only one outcome (or possibly any outcome at all) occurs during a quantum measurement \citep{dieks1989resolution, saunders1994problem, leggett2005quantum, healey2012quantum, bub2020defense, vaidman2022many, adlam2023we, karlsson2025quantum, goldstein2025bohmian}. A closely related variant, found for example in \cite{zeh1970interpretation}, \cite{fine1993measurement}, \cite{bassi2013models} or \cite{myrvold2022philosophical}, reckons that the problem arises when considering superpositions of macroscopic states; here too, the mystery seems to lie in some apparent conflict with the uniqueness of experiences at our scale. This formulation, though, can not serve as a satisfying definition either. Again, since QM is a probabilistic theory, it is normal that it involves a variety of potential outcomes, although \emph{of course} only one is to be observed by an experimentalist. In fact, if one day a physicist were to obtain several outcomes in a single experiment, or none, \emph{this} would be a much bigger mystery than what the measurement problem already is! 

At least three presentations stand out of this classification. \cite{rovelli1996relational}, aware that the measurement problem ‘assumes different faces within different interpretations’, explains that ‘a complete description of the problem can only be based on a survey of already proposed solutions’. Such a survey allows him to conclude, beyond any specific choice of interpretation, that there exists indeed an unsolved problem. He further claims that the issue arises from a misplaced use of observer-independent physical states in quantum theory, and can thus be solved by embracing a relational view on quantum systems. Despite being more accurate than the previous cited works, we will in the sequel depart from Rovelli's stance on two important aspects: first, we will argue that the problem \emph{can} be described without referring to how the different interpretations encounter it; and second, that the measurement problem is a philosophical astonishment entailed by some empirical fact about the universe, hence it is hopeless to ‘solve’ it in the usual sense of the term.

The most interpretation-independent characterisation of the problem we could find is due to \cite{maudlin1995three}, who proposed an interesting way to frame it as an incompatibility between three assumptions\footnote{In its first variant, those three assumptions are basically: (i) $\ket{\psi}$ is a complete representation of a system; (ii) the Schrödinger equation is universal; (iii) measurements have definite outcomes.}. His approach is precious insofar as it really allows to take a step back and compare or classify the different interpretations (although we will in the sequel disagree with Maudlin's conclusion that some interpretations manage to solve the problem, while others fail). One can only regret that it does not provide a \emph{positive} definition of the puzzle, so that the reader might be left with a feeling of not grasping exactly \textit{what} is weird with QM. Yes, we learn that one assertion has to be dropped, but what can we do of this? This possibly explains why, 30 years after the publication of Maudlin's article and although it is well-known in the community, his proposal did not have a major influence in the debate on the measurement problem. Notably, most of the abovementionned papers do not even cite Maudlin's definition.

Finally, \cite{wallace2016philosophy} provides an analysis of the QM formalism and concludes that ‘we are led to paradox if we think of [the quantum state] as either representational or probabilistic’, two notions that are normally well-distinguished in classical mechanics. In short, $\ket{\psi}$ can not be representational because it is unclear what a superposition of eigenstates should represent, but it can not either be treated as probabilistic because the statistics in an interference experiment does not fit with ‘the basics of the probabilistic conception’. This proposal shares some similarities with the approach that we are going to pursue in this paper, especially an effort to exhibit the patterns of quantum probabilities that distinguish them from standard probability theory. We will try to go further, though, in order to make the idea more compelling, in particular without relying on the concept of quantum state nor on any specific feature of the standard formulation of QM. 

The strength of Bell's theorem and other theory-independent no-go results resides in the fact that they would still hold and remain relevant even though QM were to be superseded tomorrow by a new theory. They tell us something about the world that physics will never forget. To achieve this, they stick to the experimental input \textit{i.e.} they merely constrain which kinds of probabilistic theories are compatible with the observed empirical statistics. In this spirit, our strategy to formulate the measurement problem in an interpretation-independent way will be to express the latter purely as a property of the empirical statistics. In section \ref{sec:kolmogorov}, we identify this property and show that it constitues the common root of the main weirdnesses of QM. We then discuss its philosophical implications in section \ref{sec:philo}, with a special focus on its relation with the notion of observer in section \ref{sec:observers}. Finally, we motivate the interest of such an interpretation-independent formulation by showing two of its possible applications: it allows to compare under a new light the different interpretations (section \ref{sec:interpretations}), and to propose a fundamental, non-technical reason why QM and general relativity (GR) are so irreconcilable (section \ref{sec:QG}).

\section{No God's-eye view on quantum matter} \label{sec:formulation}
\subsection{Quantum statistics are non-Kolmogorovian} \label{sec:kolmogorov}

Classical probabilistic phenomena are usually well described by the Kolmogorov axioms \citep{kolmogoroff1933grundbegriffe}. Let's recap them briefly. The basic mathematical landscape of probability theory is that of a measured space, that is a set $\Omega$ equipped with a $\sigma$-algebra $\mathcal{T}$ (a subset of $\mathcal{P}(\Omega)$ containing  $\Omega$ and closed under complementation and countable union) plus a map $p : \mathcal{T} \to [0,1]$, called probability measure, that satisfies the following two axioms:
\begin{itemize}
\item \textbf{Axiom 1}. \quad $p(\Omega)=1$.
\item \textbf{Axiom 2} ($\sigma$-additivity). \quad For all $(E_n)_{n\in \mathbb{N}} \in \mathcal{T}^\mathbb{N}$ disjoint, 
\[ p\left(\bigcup_{n=0}^\infty E_n\right) = \sum_{n=0}^\infty p(E_n). \]
\end{itemize}
On top of that are usually defined measurable functions from $\Omega$ to $\mathbb{R}$ called random variables. The measure $p$ allows to associate a probability for any random variable $A$ to take a particular value $a \in \mathbb{R}$ by computing $p(A^{-1}(a))$. An immediate consequence of Kolmogorov's axioms is the \emph{law of total probability} (LTP): for any pair of random variables $A$ and $B$, we have:
\begin{align} p(A=a) = \sum_b p(A=a \mid B=b) p(B=b). \label{eq:LTP} \end{align}
This equation simply follows from: $p(A=a) = p\big(\cup_b \{ A=a, B=b \}\big) =  \sum_b p(A=a , B=b)$, which translates the possibility of conditioning on all the values that $B$ can take, and to obtain the probabilities of $A$ as a marginal.

In this paper, we will specifically be concerned with \emph{physical} probabilistic theories. By ‘Kolmogorovian theory’, we will therefore mean a probability space as defined above, equipped with physically meaningful random variables, one for each quantity that can in principle be measured in the experimental situation that the theory models. For physical theories describing the mechanics of particles, this should at least include the position and momentum variables $X(t)$ and $P(t)$ at different times. Given a set of probabilities or empirical statistics, we will say that they are Kolmogorovian if there exists a Kolmogorovian theory reproducing them, and non-Kolmogorovian otherwise\footnote{Interestingly, in the case of probabilistic theories defined by a von Neumann algebra, Kolmogorovianity is equivalent to the commutativity of the algebra \citep{redei2007quantum}.}.

Now, \emph{the statistics observed in quantum experiments are non-Kolmogorovian}\footnote{It should be stressed that the statistics collected by an experimentalist after many runs of the \emph{exact same} experiment (without parameter choices) are always Kolmogorovian, because this data automatically defines a joint probability distribution for the set of random variables measured in the experiment, in agreement (by definition) with the empirical results \cite[Section 3]{szabo2007einstein}. However, the existence of a Kolmogorovian model is not guaranteed anymore if \emph{not all} variables are measured in each run.}. Indeed:
\begin{enumerate}
\item they do not satisfy the LTP, as any quantum interference confirms. For example, in a double-slit setup, one can not condition on the particle going through one slit or the other when computing the probabilities, because:
 \[ p(x) \neq p(x | \text{right}) p(\text{right}) + p(x | \text{left}) p(\text{left}).\]
In this respect, the first implementation of a quantum interference experiment (perhaps \cite{taylor1909interference}'s double-slit experiment?) was a landmark in human's history because, for the first time, it was realised that we live in a non-Kolmogorovian universe;
\item they violate the \cite{leggett1985quantum} inequality (see \cite{emary2014leggett} for a theoretical and experimental review);
\item they violate Bell-type inequalities. 
\end{enumerate}
In the following sections \ref{sec:LTP} and \ref{sec:Boole vs Bell}, we elaborate more on these two aspects of non-Kolmogorovianity that are the violation of the LTP and of the Bell inequality.

\subsubsection{The LTP and quantum (de)coherence} \label{sec:LTP}

Although non-Kolmogorovianity is a theory-independent fact about quantum matter's empirical statistics, QM correctly predicts them, so the deviation from the LTP can also be traced back in the QM formalism. When written in an orthonormal basis $\mathcal{B}=(\ket{i})_i$ (whose elements we call ‘potentialities’, as well as their corresponding events), it is well-known that the diagonal part of a density matrix accounts for the LTP when conditioning on the potentialities of $\mathcal{B}$, while the off-diagonal terms are responsible for the quantum interferences \textit{i.e.}~ the depart from the LTP. In a nutshell, this is because the probability to obtain an outcome $x$ while being in state $\rho$ is, according to the Born rule\footnote{Quantum coherence could even be defined as the ability for a physical system to deviate from the LTP. For instance, one can show that:

\[ \left\lvert \tr( \rho \Pi_{\hat{A}=a} ) -  \sum_b \tr( \Pi_{\hat{B}=b} \rho \Pi_{\hat{B}=b} \Pi_{\hat{A}=a} ) \right\rvert \leq \mathrm{rank}( \Pi_{\hat{A}=a}) \; C_{\mathcal{B}_{\hat{B}}}(\rho), \]
where $\Pi_{\hat{A}=a}$ and $\Pi_{\hat{B}=b}$ are eigenprojectors of $\hat{A}$ and $\hat{B}$ associated to the eigenvalues $a$ and $b$, and $C_{\mathcal{B}_{\hat{B}}}(\rho)$ denotes the 2-norm of coherence of $\rho$ in the eigenbasis $\mathcal{B}_{\hat{B}}$ of $\hat{B}$ \citep{soulas2025quantifying}.}:

 \[ \tr(\rho \! \ket{x}\bra{x}) = \sum_{i,j=1}^n \rho_{ij} \braket{j \vert x}  \braket{x \vert i} = \sum_{i=1}^n \underbrace{\rho_{ii}  \lvert \braket{x \vert i} \rvert^2}_{p(i) p(x \mid  i)}  +  \sum_{1\leq i < j \leq n}  \underbrace{2 \mathrm{Re}( \rho_{ij}  \braket{j \vert x} \braket{x \vert i })}_{\text{interferences}}. \]
Consequently, \emph{decoherence in a basis $\mathcal{B}$ can be understood as the process after which the LTP (approximately) holds} when conditioning on the potentialities of $\mathcal{B}$. 

The crucial importance of the LTP in QM has long been recognised. In fact, the latter is ubiquitous in the quantum foundations literature, where it is often used to translate mathematically an idea of objectivity or, more precisely, of variables having absolute values independently of being observed. For instance, the very definition of a hidden variable is an unknown quantity $\lambda \in \Lambda$ having probability distribution $\mu(\mathrm{d}\lambda)$, with respect to which any observable $A$ can be conditioned so as to satisfy the LTP:
\[ p(A) = \int_\Lambda p(A \vert \lambda) \mu(\mathrm{d}\lambda), \]
where $p(A \vert \lambda)$ is less indeterministic than $p(A)$, the extreme case being a Dirac distribution $\delta(A=f(\lambda))$. This property is already present in Bell's seminal papers \cite[equation (2)]{bell1964einstein} \cite[equation (12)]{bell1981bertlmann}, as well as in the more modern formulation of Bell's theorem in terms of causal models \cite[equation (13)]{wood2015lesson}. It is also at the basis of the ontological models program \cite[Definition 1]{harrigan2010einstein} \cite[equation (7)]{leifer2014quantum}. Similarly, the assumption of Macroscopic Realism, used to derive the Leggett-Garg inequalities, translates mathematically as a LTP \cite[equation (1)]{leggett1985quantum} and the assumption of Absoluteness of Observed Events, on which relies the Local Friendliness inequality \citep{bong2020strong}, is among other things the assumption of a LTP when conditioning on the events observed by the friends. Furthermore, the very motivation for introducing the decoherence functional in the consistent histories framework is to quantify the deviation from the LTP given a certain set of successive events \cite[equation (2.17)]{griffiths1984consistent}. More recently, \cite[equation (2)]{di2021stable} have remarked that, in relational QM, the possible violation of the LTP can be understood as a mark of the relativity of facts, while QBists have been interested in re-expressing the Born rule in order to study how closely it can resemble the LTP \cite[equation (12)]{fuchs2023qbism}.

\begin{remark}
\textbf{On Schrödinger's cat.} Rigorously speaking, the only thing that can be said to describe Schrödinger's cat puzzling situation is: ‘the statistics of a measurement of the cat in a complementary basis would not satisfy the LTP when conditioning on the cat being alive or dead’. Asserting that the cat is both alive and dead is clearly a misleading abuse, given that QM unsurprisingly predicts a $\frac{1}{2}$ probability to observe the cat either alive or dead when opening the box. Admittedly, though, this widespread expression is of course not totally ungrounded. Indeed, the LTP \eqref{eq:LTP} can be interpreted as (or at least is compatible with) the fact that one can sum independently the contribution of each potentiality $\{B=b\}$, even though $B$ is not observed, because one—and exactly one—occur each time the experiment is reproduced. Alternatively, a violation of the LTP may give a certain taste of having all the different potentialities apparently contributing to the statistics, in each single run.
\end{remark}

\begin{remark}
\textbf{Why is classical optics a wave theory?} The fundamental reason why light, ultimately made of photons, apparently behaves as a wave is the following. Since photons have no electric charge, they interact and decohere very little\footnote{See for instance the \cite{hanbury1979test} effect, employing quantum superpositions of different photon paths on astronomical scales in order to measure the size of stars!}. As bosons, they are in addition not subject to Pauli's exclusion principle. Consequently, a ‘classical’ optics experiment is best understood as many single-photon experiments in parallel, in which each photon violates the LTP. Because of this, all the different potentialities of each photon (in this case: the different optical trajectories) seem to contribute to the observed statistics, leading after application of the law of large numbers to an emergent classical wave behaviour. Besides, the supposedly ‘classical’ fact that 2  consecutive polarisers oriented as $\theta_1 = 0^{\circ}$ and $\theta_2 = 90^{\circ}$ does not transmit a laser beam, whereas adding a supplementary polariser—a device than can only \emph{absorb} energy—between them with $\theta_3 = 45^{\circ}$ \emph{does} let some photons passing through, is another illustration of the violation of the LTP in optics.
\end{remark}

\subsubsection{Boole and Bell} \label{sec:Boole vs Bell}

As a preliminary remark for this section, note that any Kolmogorovian theory trivially admits a hidden variable model: the sample $\omega \in \Omega$ is obviously a hidden variable, so that any random variable $A$ possesses in fact a definite value $A(\omega)$ in each trial. This is the core intuition of Kolmogorovian probabilities. 

Now, although he obviously did not use these terms, \cite{boole1862theory} remarked more than 150 years ago that Kolmogorovian probabilities are constrained to form a polytope. He called this set of inequalities ‘conditions of possible experience’ \citep{pitowsky1994george}. One could thus retrospectively attribute to Boole the following theorem (even though he certainly did not prove this particular statement) whose proof is straightforward and reproduced in the annex \ref{section:proof Boole}.

\begin{theorem*}[‘Boole’]
Let $A_1, A_2, B_1, B_2$ be four $\pm1$-valued random variables. Then:
\[ \lvert \langle A_1 B_1 \rangle + \langle A_1 B_2 \rangle + \langle A_2 B_1 \rangle - \langle A_2 B_2 \rangle \rvert \leq 2 \]
\end{theorem*}

\noindent The reader will have immediately recognised the CHSH inequality. Importantly, though, it does not stem here from an assumption about locality, but only as a basic property of $\pm1$-valued Kolmogorovian random variables.

A century later, \cite{bell1964einstein} (resp. \cite{bell1976theory}) showed that deterministic (resp. factorizable stochastic) hidden variable models generate probabilities which are constrained by some kind of inequality like CHSH, where:
\begin{itemize}
\item ‘deterministic’ means that the hidden variable entirely determines the values of the physical observables \textit{i.e.} $p(A =a \vert \lambda) \in \{0,1\}$ or equivalently $A = f(\lambda)$;
\item ‘factorizable stochastic’ merely assumes that $p( A = a , B = b \vert \lambda ) = p( A = a \vert \lambda ) p( B = b \vert \lambda )$.
\end{itemize}

At that point, it seems natural to ask whether Bell, despite introducing those complex assumptions about hidden variable theories, has proved anything more than the elementary theorem above. In other words, does Bell's theorem rule out a larger class of probabilistic theories candidate to reproduce the quantum statistics than Boole's theorem?

This question has been answered by \cite{fine1982hidden}. He proved that, if $A_1, A_2, B_1, B_2$ are four $\pm1$-valued physical observables, and if we are given a table of statistics (the ‘probabilities of the experiment’) obtained after performing some repeated measurements of all $A_i$, $B_j$ and all pairs $(A_i, B_j)$, the following are equivalent: (i) there exists a deterministic hidden variable model reproducing those statistics; (ii) there exists a factorizable stochastic hidden variable model reproducing those statistics; (iii) those statistics satisfy the CHSH inequality; and (iv) those statistics are Kolmogorovian \textit{i.e.}~there exists a joint probability distribution for the four observables. Consequently, this means that, mathematically speaking, Bell's assumptions are nothing but a very complicated way to re-express Kolmogorovianity. 
At the strict level of the statistical predictions, the probabilistic theories ruled out by Bell's theorem to reproduce quantum physics are exactly those ruled out by Boole's: the Kolmogorovian theories\footnote{At the level of \emph{physical} theories, though, Bell’s theorem tells us something more than Boole’s. Indeed, our theories are not a bare collection of variables, but the latter refer to entities that exist and interact within a spacetime. Bell has shown that quantum non-Kolmogorovian statistics can not even be reproduced by a local mechanism, based for example on a conspiratorial particle carrying to Bob the information of Alice's choice. This is why, to \cite{fine1982hidden} claiming: ‘hidden variables and the Bell inequalities are all about (...) imposing requirements to make well defined precisely those probability distributions for noncommuting observables whose rejection is the very essence of quantum mechanics’, \cite{shimony1984contextual} commented in these words: ‘he has, in my opinion, neglected the diversity of hidden variables theories and the diversity of the aims of different programmes. (...) the results of research on hidden variables theories are too intricate to be characterized accurately by his condensed statements.’}. It also means that Kolmogorovianity is equivalent, probabilistically speaking, to the existence of underlying hidden variables, hence Kolmogorovianity appears again as the mathematical implementation of the idea of an absolute ontology independent of epistemology.

\subsubsection{Formulation of the measurement problem}

As a conclusion of sections \ref{sec:LTP} and \ref{sec:Boole vs Bell}, we now see that non-Kolmogorovianity is associated with some idea of non-absoluteness, and that it constitutes the root of the two main weirdnesses of QM:
\begin{enumerate}
\item superposition/coherence/interferences \textit{via} the violation of the LTP;
\item entanglement/Bell non-locality \textit{via} the violation of Bell-type inequalities.
\end{enumerate}
We therefore propose the following interpretation-independent formulation: \\
\begin{center}
\emph{the measurement problem is the fact that quantum matter\\ cannot be described by a Kolmogorovian theory}.
\end{center}

\noindent \\  In the next section, we examine the philosophical consequences of this definition.

\begin{remark}
\textbf{Why has it been called the ‘measurement’ problem?} The more one thinks about it, the less it seems to have anything to do with measurements... The only explanation we could think of is historical, due to the primacy of the Copenhagen interpretation, whose most obvious weakness lies in the collapse, which was associated with the act of measurement. A better phrasing would certainly be ‘non-Kolmogorovianity problem’ or ‘quantum non-absoluteness’. 
\end{remark}

\subsection{The philosophical significance of non-Kolmogorovianity} \label{sec:philo}

Our key argument can be cast into the standard form of theory-independent no-go results that \cite{shimony1984contextual} called ‘experimental metaphysics’ (see also \cite{cavalcanti2008reality}), albeit in a very simplistic version. Consider the following metaphysical assumption (\textsc{epistemic uncertainty}):

\paragraph{\textsc{(EU)}} All uncertainty is epistemic: even though the actual value of a variable is not known, it still has a single definite value among the possible ones.\\

\noindent Motivated by section \ref{sec:kolmogorov}, we claim that \emph{Kolmogorov axioms are the axioms of epistemic uncertainty}; in other words Kolmogorovian probability theory constitutes the mathematical implementation of \textsc{(EU)}. Experimentally, we observe that the universe’s statistics violate these axioms. Hence we must reject \textsc{(EU)}: this is the main philosophical consequence of the measurement problem.

Let's try to reformulate what the rejection of \textsc{(EU)} entails philosophically. We have already seen in section \ref{sec:kolmogorov}, through considerations on the LTP and on the equivalence between Kolmogorovianity and hidden variables, that Kolmogorov axioms are associated with a notion of ontology independent of epistemology\footnote{The link between our proposed formulation of the measurement problem and the usual textbook presentation (the conflict between unitary evolution and collapse) is the following. Any probabilistic theory naturally involves an epistemic update of the probabilities when new information is acquired. However, in a non-Kolmogorovian theory where ontology can not be clearly separated from epistemology, this supposedly epistemic operation takes a puzzling ontic flavour. It may therefore be tempting to consider it as as genuine physical process, granted with the same status as the time evolution law, but doing so leads to apparently ill-defined conditions of applicability.}. The negation of \textsc{(EU)} can be further rephrased as the impossibility to build such an absolute ontology for quantum matter, supporting \cite{zurek2022emergence}'s intuition: ‘I strongly suspect that the ultimate message of quantum theory is that the separation between what exists and what is known to exist – between the epistemic and the ontic – must be abolished.’ Said differently, there exists no God's-eye view on quantum systems, no ‘view from nowhere’ \citep{nagel1989view, latour2017facing}. By God's-eye view, we mean a description of the entire universe in a well-definite configuration, perhaps inaccessible to actual observers but true for everyone. It depicts ‘all there is’, so that in particular every internal, partial perspective can be derived from it—the paradigmatic example being of course classical physics. 

In the history of science, the fact that the measurement problem was deemed a problem may just be the symptom of a physics that took itself too much for God. Isn't the measurement ‘problem’, indeed, only an issue for physicists too reluctant to abandon the God's-eye view? Not exactly. In response to someone arguing: ‘I have accepted that the universe is non-Kolmogorovian and I feel comfortable with the absence of absolute ontology; at the end of the day, is there any problem?’, two questions can be put forward.
\begin{enumerate}
\item \textbf{The physical problem:} how to explain that we do not experience non-Kolmogorovianity at our scale, to such an extent that no human ever had any reason to doubt of the applicability of Kolmogorov axioms in physics before the advent of QM? This is a central aim of the research on the quantum-to-classical transition; see for instance \cite{gell1993classical, zurek2003decoherence, kofler2007classical} among many others.
\item \textbf{The philosophical problem:} which ontology should we adopt in a non-Kolmogorovian universe? This is the question that every interpretation of QM is trying to address.
\end{enumerate}
The former problem clearly lies outside the scope of this paper. In the following sections \ref{sec:observers} and \ref{sec:interpretations}, we comment on the latter.

\subsubsection{Observers: when the philosophical difficulty worsens} \label{sec:observers}

A growing number of works in the field of quantum foundations focus on introducing the 1st person point of views in the description, as exemplified by the recent surge of interest in Wigner's friend thought experiment \citep{schmid2023review}, by the rise of quantum reference frames \citep{angelo2010physics, giacomini2019quantum}, as well as a myriad of other approaches \citep{fuchs2016participatory, frauchiger2018quantum, mueller2020law, fankhauser2025epistemic}. This should not come as a surprise once it is understood that, in a non-Kolmogorovian universe, ontology can not be independent of epistemology, as argued above. The view from nowhere must be replaced by \textit{in situ} perspectives.

Wigner's friend precisely captures a philosophical tension that makes it harder—although not impossible—to build a satisfying ontology in presence of non-Kolmogorovianity. Ask yourself: what is so special when a human being  (or possibly a cat) is inserted into the description, instead of particle? Manifestly, no one was shocked by the experimental violation of the Local Friendliness inequality when Charlie and Debbie were taken to be single photons \citep{proietti2019experimental, bong2020strong}. So what is more troubling with an actual human being? We propose the following: the philosophical problem becomes even more severe when the system that deviates from Kolmogorovianity is itself an observer, which we define here as \textit{any physical system granted with subjectivity}, because of the tension between: (i)~the multiplicity of potentialities apparently contributing to the statistics when the LTP is violated, and (ii)~the unicity of the subjective experience. \cite{wigner1995remarks} insisted that ‘my friend [should have] the same types of impressions and sensations as I’, in particular she should have unique experiences, so how come the other potentialities still contribute to the statistics when Wigner performs a supermeasurement on her? In a non-Kolmogorovian ontology, even 1st person perceptions can not be absolute.

An important related question is of course to determine which systems can be treated as observers. This is a currently hot debate: \cite{fields2012if} has for exemple proposed to define an observer through its information-processing abilities, when \cite{brukner2021qubits} has argued that ‘qbits are not observers’. \cite{stoica2025makes} has defended the idea that structures alone are insufficient to determine the sentient observers of quantum theory, while \cite{wiseman2023thoughtful} have discussed the validity of the Friendliness assumption, granting a system with ‘thoughts (...) as real as any communicable thought of my own’, for instance in the case of an AI running on a quantum computer.

\section{Applications}
\subsection{Comparing the main interpretations} \label{sec:interpretations}

We have now exposed the difficulties faced by anyone willing to interpret QM \textit{i.e.}~endeavouring to build an ontology compatible with what we know of quantum matter. A virtue of our interpretation-independent formulation of the measurement problem is that it allows to identify \textit{a posteriori} the common question that each interpretation actually attempted to address, and thereby to show them under a new light. In full generality, \textit{one could define an interpretation of QM as an attitude towards non-Kolmogorovianity}. Below is a snapshot of the main attitudes that have been proposed so far.

\begin{itemize}
\item \textbf{Copenhagen.} Physics is first and foremost an activity of real humans in actual labs with virtually infinite resource: a quantum state $\ket{\psi}$ does not make sense out of such a context. One must remain agnostic about the ‘true nature of the electron’, hence the latter does not need to be Kolmogorovian.

\item \textbf{Collapse models.} A modified Schrödinger dynamics causes an objective collapse (here a genuine physical process, not an epistemic update) for large systems, which become Kolmogorovian in the preferred basis of position. Before collapse, the stance on non-Kolmogorovianity may vary from one author to the other.

\item \textbf{Bohmian mechanics.} Non-Kolmogorovianity is repressed: Bohmian mechanics (BM) is the ultimate attempt to save the God’s-eye view. The ontology reduces to a collection of particles with definite positions and velocities. Despite absolute trajectories, though, the non-Kolmogorovian quantum statistics are successfully reproduced (see remark \ref{rk:BM} below).

\item \textbf{Many-worlds.} As an attempt to save a meta-God’s-eye view, one postulates that, ontologically, the universe \emph{is} $\ket{\psi}$. Non-Kolmogorovianity arises because all potentialities do happen and contribute to the statistics. Decoherence produces branches that satisfy the LTP.

\item \textbf{QBism.} There is no objective, outside world, but only a myriad of 1st person point of views. Each agent has their own $\ket{\psi}$, merely encoding their degrees of belief, and non-Kolmogorovianity becomes unproblematic as soon as measurement outcomes cease to be considered as objective features of the world.

\item \textbf{Relational quantum mechanics.} Any physical interaction between two systems may establish a fact for them. However, a linguistic effort is required: one mustn't talk about the value of a variable \emph{per se}, but only \emph{relative} to another system which has decohered it. So, whenever the LTP is violated, the conditioning was made on something which is not a fact for anyone: this is not permitted.
\end{itemize}

\begin{remark} \label{rk:BM}
\textbf{Does Bohmian mechanics contradict our argument?} 
Because it reproduces the quantum statistics, BM is certainly not a Kolmogorovian theory. The fact that it involves a probability distribution over the positions of particules $x(t)$, and that the velocities are defined as $v(t) = \dot{x}(t)$, may give the feeling that BM has the characteristics of a Kolmogorovian model as we defined it in section \ref{sec:kolmogorov}. This is not correct, though: in BM, only the position variable (and functions thereof) is a genuine physically meaningful random variable. Studying the random variable $p(t) = mv(t)$ fails to yield the correct probabilities for a momentum measurement. To predict these, it is instead necessary to include a full description of the measurement apparatus' dynamics, in which the measurement outcome is ultimately encoded, as everything else, in terms of positions of particles (in this case, in the position of the apparatus' pointer). Consequently, there exists no random variable describing momentum in BM, nor spin.
For the same reason, although uncertainty in BM might at first glance seem purely epistemic, the theory actually rejects \textsc{(EU)}, in accordance with our claim of section \ref{sec:philo}, because momentum or spin do not have any value at all: it is simply meaningless to talk about them. Let's mention in passing that the situation is quite similar for \cite{spekkens2007evidence}'s toy model: it is a non-Kolmogorovian theory (it reproduces \textit{e.g.}~quantum interference, hence violates the LTP) that rejects \textsc{(EU)} since measurable quantities do not have a pre-existing value in general (the ontic state merely dictates their probabilities). The same can be said of \cite{barandes2025stochastic} indivisible stochastic approach: as a reformulation of QM, it is non-Kolmogorovian; and because the configuration space corresponds in the Hilbert space to the vectors of a preferred-basis, only one set of commuting observables is assigned well-definite values, hence \textsc{(EU)} is not satisfied.

In defense of BM, one may insist and argue that the theory proposes an ontology in which \emph{only} position exists, so it could seem unfair to demand it to provide random variables with well-defined underlying values for the other variables as well. Leaving aside the conspiratorial flavour of such an ontology\footnote{Considering that we are able to build such a wide variety of measurement devices probing plenty of different physical quantities which exhibit strong regularities, and among which position does not dramatically stand out, is it reasonable to believe that all these other variables are in fact mere hallucinations caused by patterns in the behaviour of the positions of the devices' particles?}, even the epistemic nature of the position uncertainty is questionable in BM. Indeed, one can ask why, in spite of being a deterministic theory, it will ultimately never have more predictive power than QM \citep{colbeck2011no}. It is because perfect knowledge (or simply more knowledge than what is given by the Born rule $\lvert \psi(x) \rvert^2$) of initial conditions of Bohmian particles is in principle inaccessible, a fact known as the quantum equilibrium hypothesis (QEH), necessary to ensure no-signalling \citep{valentini2002signal}. The problem is that such a hypothesis is not easily justified. Heisenberg microscope-type arguments may well qualitatively explain an irreducible epistemic uncertainty, but they fail to quantitatively reproduce the quantum uncertainties where the $\hat{X}$ and $\hat{P}$ distributions are precisely related by a Fourier transform \citep{brown1981critique}. Moreover, although \cite{durr1992quantum} have shown that the uncertainty tends to converge to the QEH, it remains unclear why it is in principle impossible to reduce one's epistemic uncertainty by more than $\lvert \psi(x) \rvert^2$, be it for a short time, be it in a single non-typical experiment. Consequently, considering how this irreducible Bohmian uncertainty seems to need to be imposed by \textit{fiat}, one may wonder to what extent it can accurately be deemed ‘epistemic’. Interestingly, \cite{durr1992quantum} themselves call it ‘absolute uncertainty’. Furthermore, BM's trajectories have a gauge indeterminacy because the probability current $\vec{j}$ guiding the particles is only defined up to a divergence-free field \citep{fankhauser2025undetectability}, since it is constrained solely by the equation $\partial_t \lvert \psi \rvert^2 + \mathrm{div}(\vec{j}) = 0$. Although this gauge under-determinacy makes no difference at the level of the predictions, BM still does not tell us which trajectory is the actual one, hence it does not really let us see what God's eye sees.

\end{remark}

\subsection{Why is quantum gravity so hard?} \label{sec:QG}

In addition to its interpretational benefit, the proposed formulation of the measurement problem might also be fruitful in the context of quantum gravity (QG). Indeed, as we shall argue now, \emph{formulating why QM is so weird helps understanding why QG is so hard}. It turns out that it is not easy to exhibit a compelling reason for why QM and GR are so incompatible, without relying on purely technical, mathematical arguments (\textit{e.g.}~non-renormalisability) and without presupposing the form of the sought theory (\textit{e.g.}~quantized metric field, superpositions of spacetimes...).  Let's therefore start with this astonishing observation: how come Einstein was able to formulate a backgroundless theory of classical matter as soon as 1915, and no one has been able to do so for quantum matter for more than a century? 

Einstein's theory of relativity is deemed relativistic because lengths and durations are relative to the observer. However, his ‘relativistic’ spacetime is still very absolute. Mind that one possible formulation of Einstein's principle of relativity is the requirement that the laws of physics should take the same form in all admissible reference frames: shouldn't it rather be called the principle of absoluteness? In fact, \cite{letterEinsteinZshimmer} himself was unsatisfied with the name ‘relativity’ and would have preferred ‘invariance theory’. This flavour of absoluteness can be traced back in the mathematics of the theory: in differential geometry, although the charts account for the 1st person perspectives, \emph{the manifold inevitably pictures a God's-eye view}. In essence, a manifold is an absolute structure: true for everyone, containing all the geometrical information possessed by the different 1st person perspectives, and from which any of them can be derived—something that quantum matter can not admit, as we have seen. Now, as long as they remain decoupled, an absolute spacetime can very well host some quantum matter, as is the case in quantum field theory. However, in a theory like GR where spacetime and matter co-emerge, and where they get equated \textit{via} Einstein equation, the absoluteness of the former would clash with the non-absoluteness of the latter, if it were quantum. Here is, we propose, the fundamental tension between QM and GR.

The challenge of QG could then be: which mathematics can supersede differential geometry, and describe spacetime without allowing for any God's-eye view? Which mathematical structure could account for the non-Kolmogorovianity of the spacetime variables? Superposition of classical spacetimes as vectors in a Hilbert space might be a way to do so, but a very restrictive one. In particular, it only allows for a global violation of the LTP at the level of the entire spacetime, but not for a locally-dependent violation. 

It could be argued that fixing a bare topological manifold almost doesn't fix anything geometrically, and that the non-absoluteness of spacetime's geometry could be accounted for by simply upgrading the metric field to an operator field, as is done in standard approaches to quantize GR. However, such a structure already presupposes \emph{the absoluteness of the notion of event}, namely the fact that there exists an objective collection of points, common to all, each of which can be referred to by everyone, and where something happens. QG invites us to question this: when is it admissible to say that something actually does happen, and where, and who can say so? This is precisely the kind of questions approached by the field of quantum causality \citep{brukner2014quantum, goswami2020experiments}, for instance with the study of indefinite causal orders, process matrices or the quantum switch. 

\section{Conclusion}

Here is what quantum mechanics (QM) tells us about the world: the universe is not Kolmogorovian. There is no way to ‘solve the measurement problem’ simply because it is an empirical fact. This poorly-named problem is neither a physical nor a mathematical problem, but it motivates a philosophical one: the task of building an ontology compatible with the experimental violation of Kolmogorov axioms. Like any other philosophical question, this does not admit a definite, universal solution; in fact, there are as many answers as there are interpretations of QM. On the one hand, none of them succeed to solve the measurement problem, in the sense that they do not change how the empirical statistics behave; on the other hand, all of them succeed in proposing an interesting attitude towards non-Kolmogorovianity.

While the measurement problem does not have a universal solution, it \emph{does} however admit a definition, and formulating it in an interpretation-independent way could be fruitful for physics. Once the common root of the puzzle has been identified, one can hope for a better dialogue between the different schools of interpretations of QM. It might also guide the search for quantum gravity, by highlighting the deep source of tension between QM and general relativity. 

The God’s‑eye view has been one of the pillars of modern physics since the scientific revolution \citep{daston2021objectivity}. It lies at the heart of classical physics, but also of Einstein’s relativity, where the spacetime manifold is seen from nowhere. It is the lesson of the measurement problem, and the challenge of quantum gravity, to finally get rid of it. The quest against absoluteness is still ongoing.

\section*{Acknowledgements}

During the course of this 5-years long project, countless discussions with colleagues, multiple feedbacks received from re-readers, and sharp questions received in seminars and conferences were the main driving forces to gradually improve the quality of this article over the successive rewritings. The author would like to thank in particular Emily Adlam, Luis C.~Barbado, Veronika Baumann, Časlav Brukner, Lin-Qing Chen, Dmitry Chernyak, Michael Cuffaro, Maxime Devin--Pinson, Andrea Di Biagio, Gyözö Egri, Johannes Fankhauser, Guilherme Franzmann, Simon Fuchs, Alexei Grinbaum, Sebastian Horvat, Charlène Laffond, Baptiste Le Bihan, Niels Linnemann, Antoine Pinochet--Lobos, Carlo Rovelli, Alberto Spalvieri, László E.~Szabó, Anderson A.~Tomaz, Lev Vaidman, David Wallace, Christian Wüthrich and his students, as well as Anton Zeilinger.

This research was funded by the Austrian Science Fund (FWF) [10.55776/F71] and [10.55776/COE1], and by BeyondC [10.55776/RG3] (Research Group 3). For open access purposes, the authors have applied a CC BY public copyright license to any manuscript version arising from this submission. This publication was made possible through the financial support of the WOST (WithOutSpaceTime) grant from the John Templeton Foundation. The opinions expressed in this publication are those of the authors and do not necessarily reflect the views of the John Templeton Foundation.

\begin{appendices}
\section{Appendix}
\subsection{Proof of ‘Boole's theorem’} \label{section:proof Boole}
We present here the proof of the theorem introduced in \ref{sec:Boole vs Bell} and somewhat cheekily attributed to Boole.

\begin{theorem*}[‘Boole’]
Let $A_1, A_2, B_1, B_2$ be four $\pm1$-valued random variables. Then:
\[ \lvert \langle A_1 B_1 \rangle + \langle A_1 B_2 \rangle + \langle A_2 B_1 \rangle - \langle A_2 B_2 \rangle \rvert \leq 2 \]
\end{theorem*}

\begin{proof}
Notice that $\lvert A_1 \rvert = \lvert A_2 \rvert = \lvert B_1 \rvert = \lvert B_2 \rvert = 1$, and compute:
\begin{align*}
\lvert \langle A_1 B_1 \rangle + \langle A_1 B_2 \rangle + \langle A_2 B_1 \rangle - \langle A_2 B_2 \rangle \rvert 
=& \lvert \langle A_1 (B_1 +B_2) \rangle + \langle A_2 (B_1 - B_2) \rangle \rvert \\
\leq& \Big\langle \lvert A_1 \rvert \lvert B_1 +B_2 \rvert \Big\rangle + \Big\langle \lvert A_2 \rvert \lvert B_1 - B_2 \rvert \Big\rangle \\
\leq& \Big\langle \lvert B_1 +B_2 \rvert + \lvert B_1 - B_2 \rvert  \Big\rangle \\
\leq& 2 \Big\langle \max( \lvert B_1 \rvert , \lvert B_2 \rvert  ) \Big\rangle \\
\leq& 2,
\end{align*}
where we have used the fact that $\forall x,y \in \mathbb{R},\; \lvert x+y \rvert + \lvert x-y \rvert \leq 2\max(\lvert x \rvert, \lvert y \rvert)$. 
\end{proof}
\end{appendices}

\newpage
\bibliographystyle{plainnat}
\bibliography{Biblio_interpretation}
\end{document}